\documentclass[conference]{IEEEtran}
\IEEEoverridecommandlockouts

\usepackage{cite}
\usepackage{amsmath,amssymb,amsfonts}
\usepackage{amsthm}
\usepackage{algorithmic}
\usepackage{algorithm}
\usepackage{graphicx}
\usepackage{textcomp}
\usepackage{xcolor}
\usepackage{booktabs}
\usepackage{multirow}
\usepackage{hyperref}

\newtheorem{proposition}{Proposition}
\newtheorem{theorem}{Theorem}
\newtheorem{corollary}{Corollary}

\def\BibTeX{{\rm B\kern-.05em{\sc i\kern-.025em b}\kern-.08em
    T\kern-.1667em\lower.7ex\hbox{E}\kern-.125emX}}

\begin{document}

\title{Priority-Aware Routing for Quantum Networks:\\
Integrating Coherence-Time Constraints into Scheduling}

\author{
\IEEEauthorblockN{
    Sadhgun Ram Dasi\IEEEauthorrefmark{1},\quad
     Aswath Babu H\IEEEauthorrefmark{2}
}

\IEEEauthorblockA{
    \IEEEauthorrefmark{1}Department of Electronics and Communication Engineering\\
    Indian Institute of Information Technology Dharwad\\
    Dharwad, India\\
    \textit{sadhgunramdasi@gmail.com}
}

\IEEEauthorblockA{
    \IEEEauthorrefmark{2}Department of Arts, Science and Design\\
    Indian Institute of Information Technology Dharwad\\
    Dharwad, India\\
    \textit{aswath@iiitdwd.ac.in}
}
}

\maketitle

\begin{abstract}
Quantum networks face a fundamental challenge absent in classical
networks: finite memory coherence times mean that queuing delay
directly degrades information quality, causing decoherence that
has no classical analog. Existing routing protocols optimize for
hop count or channel fidelity while treating scheduling delays
as a secondary concern, leaving fidelity vulnerable to
congestion-induced decoherence.

We present a coherence-aware routing protocol that integrates
decoherence constraints into path selection via a
priority-modulated aging model derived analytically from
Weighted Fair Queueing (WFQ) theory, where high-priority
traffic receives compressed effective queue delay. We build
a custom discrete-event simulator---verified against NetSquid
($\mathrm{MAE} < 10^{-6}$ across $10^4$ states)---and
evaluate across two structurally diverse topology classes under
identical physical parameters, isolating topology structure as
the sole independent variable. In all experiments, our protocol
is benchmarked against loss-only Dijkstra routing and FIFO
scheduling as baselines.

On Erd\H{o}s--R\'{e}nyi topologies, $P_2$ traffic achieves
$96.08\%$ fidelity under a $9\times$ load increase with only
$0.04$ percentage point degradation, versus a $13.51$ point
collapse for the loss-only baseline; $P_2$ tail latency holds
stable at $0.055$~ms while Dijkstra degrades $5.5\times$.
On Barab\'{a}si--Albert scale-free topologies, our protocol
delivers a $+4.12$ percentage point fidelity advantage with
$97.8\%$ higher aggregate throughput at operational loads
over the baseline, with a clearly characterized hub-saturation
boundary at 160k~QPS. FIFO scheduling provides negligible
priority differentiation on both topology classes, providing
strong evidence that coherence-aware routing---not scheduling
alone---is responsible for the observed QoS gains.
\end{abstract}


\section{Introduction}
\label{sec:intro}

\IEEEPARstart{T}{he} development of quantum internet infrastructure
relies on the ability to distribute entanglement and transmit
quantum information across remote nodes with sufficient
fidelity~\cite{kimble2008quantum, wehner2018quantum}. Applications
ranging from distributed quantum computing to quantum key
distribution require the transmission of qubits over channels that
may introduce loss and through intermediate nodes with finite
storage capabilities. Unlike classical networks where data can
typically be buffered without degradation, quantum networks must
address the store-and-decay problem.

Quantum memories---whether based on trapped
ions~\cite{wineland2013quantum}, nitrogen-vacancy (NV) centers in
diamond~\cite{doherty2013nitrogen}, or superconducting
circuits~\cite{kjaergaard2020superconducting}---lose coherence over
time due to environmental interactions~\cite{preskill1998quantum}.
This decoherence can manifest as dephasing, amplitude damping, or
depolarization, characterized by timescales $T_1$ (energy
relaxation) and $T_2$ (phase coherence)~\cite{nielsen2010quantum}.
A notable characteristic of quantum networks is the relationship
between time and information quality: delays from routing
decisions, scheduling, or network congestion can lead to reduced
fidelity and potential information loss.

\subsection{Motivation and Challenge}

Existing quantum routing protocols typically optimize for hop
count~\cite{schoute2016shortcuts} or entanglement generation
rate~\cite{pant2019routing}, implicitly treating scheduling
delays as a secondary concern. Under congestion, however, packets
traversing even optimal paths accumulate queue waiting time that
directly erodes fidelity through decoherence.

Classical Quality of Service (QoS) mechanisms---such as priority
queuing~\cite{demers1989analysis, tanenbaum2011computer} or
weighted fair queuing~\cite{parekh1993generalized}---typically
operate at the scheduling layer. Recent work has explored quantum network protocol design~\cite{kozlowski2020designing}
and congestion-aware entanglement routing~\cite{shi2020concurrent},
though priority is typically treated as a logical label rather
than a parameter that influences path selection.

This raises the question: how might routing protocols be designed
to account for quantum states that degrade during transmission,
and does such a protocol generalize across structurally diverse
network topologies?

\section{Related Work}
\label{sec:related}

Fidelity-aware quantum routing has been explored through
entanglement generation rate optimization~\cite{pant2019routing},
graph-theoretic path selection~\cite{schoute2016shortcuts}, and
MDP-based approaches~\cite{caleffi2017optimal}, but none model
time-dependent decoherence during transit.
Chakraborty et al.~\cite{chakraborty2019distributed} treat
priority as a logical label rather than a routing cost; we make it
a first-class variable through the WFQ-grounded aging model.
At the link layer, Dahlberg et al.~\cite{dahlberg2019link} treat
scheduling and routing as independent concerns, while our work
integrates them. On the queuing side, Vardoyan et
al.~\cite{vardoyan2019stochastic} provide stochastic foundations
for quantum switch behavior that our cost model builds on.

From classical networking, our priority classes mirror
DiffServ~\cite{blake1998architecture} and the aging model is
formally grounded in WFQ~\cite{demers1989analysis,
parekh1993generalized}. The key departure is that in quantum
networks, service differentiation must happen at the routing
layer---path selection determines the decoherence budget, a
dependency with no direct classical analogue.

\subsection{Contributions}

This paper makes the following contributions:

\begin{enumerate}
    \item \textbf{WFQ Aging Model:} We propose a
    priority-modulated aging model formally derived from Weighted
    Fair Queueing theory~\cite{demers1989analysis,
    parekh1993generalized}, where the $(1+P)^{-1}$ scaling factor
    is analytically justified as a first-order approximation of
    WFQ-induced delay compression.

    \item \textbf{Integrated Protocol Design:} We implement a
    routing cost function that considers predicted queue waiting
    times and coherence-time constraints, allowing paths to be
    selected based on both channel quality and temporal
    feasibility. We establish formal complexity bounds
    (Proposition~\ref{prop:complexity}) and a minimum coherence
    time guarantee (Theorem~\ref{thm:coherence_bound}).

    \item \textbf{Stress Regime Analysis:} We evaluate the system
    in a stress regime ($T_{\text{mem}} = 400~\mu$s), where
    coherence times are sufficient for multi-hop routing but short
    enough that queuing delays are critical.

    \item \textbf{Rigorous Baseline Comparison:} We compare against
    FIFO scheduling, a loss-only Dijkstra baseline ($\beta = 0$),
    and classical priority queuing without coherence-aware routing,
    isolating the contribution of coherence-aware path selection.
    We demonstrate that FIFO provides \emph{zero} priority
    differentiation and that loss-only routing collapses below the
    fidelity utility threshold $F_{\min}$ under moderate load.

    \item \textbf{Two-Topology Generalization Study:} We evaluate
    the protocol on both Erd\H{o}s-R\'enyi random
    graphs~\cite{erdos1960evolution} and Barab\'asi--Albert
    scale-free graphs~\cite{barabasi1999emergence} under identical
    physical parameters, isolating topology structure as the sole
    independent variable and characterizing a topology-dependent
    operational envelope for coherence-aware routing.
\end{enumerate}

The paper is organized as follows:
Section~\ref{sec:background} reviews decoherence modeling;
Section~\ref{sec:system_model} introduces the priority-aware
system model;
Section~\ref{sec:protocol} describes the routing protocol;
Section~\ref{sec:simulator} presents the simulator architecture;
Section~\ref{sec:experiments} details experimental methodology
and results;
Section~\ref{sec:discussion} discusses observed behaviors;
Section~\ref{sec:limitations} addresses limitations; and
Section~\ref{sec:conclusion} concludes with future directions.

\section{Background: Modeling Decoherence}
\label{sec:background}

To incorporate coherence considerations into routing, we model
quantum memory noise using established methods. This section
reviews the fidelity evolution model used for packet drop decisions
and cost computation.

\subsection{Density Matrix Evolution}

Let $\rho(t)$ represent the density matrix of a qubit stored in
memory at time $t$. Following standard approaches, we model the
dominant noise as $T_2$ dephasing, characteristic of
superconducting qubits~\cite{kjaergaard2020superconducting} and
solid-state memories~\cite{doherty2013nitrogen}. The evolution is
described by the Lindblad master
equation~\cite{lindblad1976generators, gorini1976completely}:
\begin{equation}
\frac{d\rho}{dt} = \gamma \left( Z\rho Z^\dagger - \rho \right)
\label{eq:lindblad}
\end{equation}
where $\gamma = 1/(2T_{\text{mem}})$ is the dephasing rate and
$Z = \begin{pmatrix} 1 & 0 \\ 0 & -1 \end{pmatrix}$ is the
Pauli-Z operator.

For discrete-event simulation, we use the operator-sum
representation (Kraus
decomposition)~\cite{nielsen2010quantum, kraus1983states}:
\begin{equation}
\rho(t + \Delta t) = \mathcal{E}(\rho(t)) =
\sum_{k=0}^{1} K_k \rho(t) K_k^\dagger
\label{eq:kraus}
\end{equation}
with Kraus operators:
\begin{align}
K_0 &= \sqrt{1 - p(\Delta t)} \, I \label{eq:kraus0} \\
K_1 &= \sqrt{p(\Delta t)} \, Z \label{eq:kraus1}
\end{align}
where the phase-flip probability is:
\begin{equation}
p(\Delta t) = \frac{1 - e^{-\Delta t / T_{\text{mem}}}}{2}
\label{eq:phase_flip}
\end{equation}

\subsection{Fidelity Decay and Threshold}

We define fidelity with respect to the ideal initial state
$|\psi\rangle$ as $F(t) = \langle\psi|\rho(t)|\psi\rangle$~%
\cite{jozsa1994fidelity, uhlmann1976transition}. For a pure
initial state undergoing pure dephasing, the fidelity decays
exponentially:
\begin{equation}
F(t) = \frac{1 + e^{-t/T_{\text{mem}}}}{2}
\label{eq:fidelity_decay}
\end{equation}

\noindent\textbf{Remark:} Equation~\eqref{eq:fidelity_decay} holds
exactly for pure initial states under pure $T_2$ dephasing. Mixed
states or amplitude damping ($T_1$ processes) yield more complex
evolution; this is a known simplification noted further in
Section~\ref{sec:limitations}.

To maintain utility for applications such as quantum
teleportation~\cite{bennett1993teleporting}, we set a fidelity
threshold following established
guidelines~\cite{massar1995optimal}:
\begin{equation}
F(t) < F_{\min} = 0.70 \quad \Rightarrow \quad \text{Discard Packet}
\label{eq:threshold}
\end{equation}
This threshold lies above the classical fidelity bound of $2/3$
for single-qubit teleportation~\cite{massar1995optimal}, below
which quantum transmission offers no advantage over classical
communication.

\subsection{Evaluation Topologies}
\label{sec:topologies}

We evaluate the Coherence-Aware (CA) routing protocol on two
distinct synthetic network topologies, each designed to expose
different structural challenges in quantum network routing.

\subsubsection{Erdős–Rényi Random Topology}
The first topology is an Erdős–Rényi (ER) random
graph~\cite{erdos1960evolution}, used in the load-sweep (Suite~I) and
multi-seed robustness (Suite~II) evaluations. In this model, $N =
30$ nodes are placed uniformly at random over a $5\,\text{km}
\times 5\,\text{km}$ geographical area, and each pair of nodes is
connected independently with probability $p$. We use $p = 0.30$ in
dynamic evaluations and $p = 0.15$ for static graph generation.

ER graphs are statistically homogeneous: degree variance is low,
hub nodes are absent, and multiple disjoint paths of roughly equal
length typically exist between any source--destination pair. This
uniformity makes the ER topology an effective \emph{baseline},
isolating the ability of CA routing to balance traffic across
equal-cost paths before the network saturates uniformly.

\subsubsection{Barabási–Albert Scale-Free Topology}
The second topology is a Barabási–Albert (BA) scale-free
graph~\cite{barabasi1999emergence}, used in the hub-stress
evaluations (Suite~III). BA graphs are constructed via preferential
attachment: each newly added node connects to existing nodes with
probability proportional to their current degree, producing a
power-law degree distribution with a small number of highly
connected hub nodes.

This structure serves as a deliberate \emph{stress test}. Because
most shortest paths necessarily traverse the high-degree hubs, those
nodes become severe bottlenecks under elevated load---a behavior
clearly visible in the observed throughput collapse at $\gtrsim
160{,}000$ QPS in our experimental logs. The BA topology therefore
probes whether CA routing can exploit alternative, hub-avoiding
paths to preserve fidelity until no such detour remains.

Together, the ER and BA topologies provide complementary coverage:
the former tests homogeneous saturation behavior, while the latter
tests resilience under heterogeneous, hub-dominated congestion.

\section{Priority-Aware System Model}
\label{sec:system_model}

We model the quantum network as a directed graph $G = (V, E)$,
where each node $v \in V$ represents a quantum repeater with
finite quantum memory, classical control processors, and photonic
interfaces. Each edge $(u, v) \in E$ represents a bidirectional
optical link with photon loss probability $L_{uv}$ and propagation
latency $D_{\text{prop}}(u, v)$.

\subsection{Priority-Based Aging Model}

We introduce a priority-aware aging model grounded in Weighted
Fair Queueing (WFQ) theory~\cite{demers1989analysis,
parekh1993generalized}. Let $P \in \{0, 1, 2\}$ denote the
priority class, where $P = 2$ is highest priority.

\textbf{Weight Mapping.} We map each priority class to a service
weight $w_P = 1 + P$, giving:
$w_0 = 1$ (Low), $w_1 = 2$ (Medium), $w_2 = 3$ (High).

Under WFQ, each class receives service bandwidth proportional to
its weight:
\begin{equation}
R_P = R_{\text{total}} \cdot \frac{w_P}{\sum_k w_k}
    = R_{\text{total}} \cdot \frac{1+P}{6}
\label{eq:wfq_rate}
\end{equation}

\begin{proposition}[Delay Compression]
\label{prop:delay_compression}
Under a WFQ scheduler with weights $w_P = 1 + P$ and total
capacity $R_\mathrm{total}$, the expected queuing delay for
class $P$ satisfies:
\begin{equation}
    E[W_P] \;\approx\; \frac{E[W_\mathrm{FIFO}]}{1+P}
\end{equation}
with approximation error $O(\bar{L}/R_\mathrm{total})$, where
$\bar{L}$ is the mean packet size, valid in the heavy-traffic
regime $\rho \to 1$.
\end{proposition}

\begin{proof}
WFQ is a packetized approximation of Generalized Processor
Sharing (GPS)~\cite{parekh1993generalized}. Under GPS with weights
$w_P$, class $P$ receives service bandwidth:
\begin{equation}
    R_P \;=\; R_\mathrm{total} \cdot \frac{w_P}{\sum_k w_k}
    \;=\; R_\mathrm{total} \cdot \frac{1+P}{6}
\end{equation}
where $\sum_k w_k = w_0 + w_1 + w_2 = 1 + 2 + 3 = 6$.
Applying the Kleinrock conservation law~\cite{kleinrock1976queueing}
to an M/G/1 queue under work-conserving scheduling, the
mean unfinished work $U$ is invariant across scheduling
disciplines:
\begin{equation}
    \sum_P \rho_P \, E[W_P] \;=\; \text{const}
\end{equation}
where $\rho_P = \lambda_P / \mu_P$ is the load contributed
by class $P$. Under symmetric load ($\rho_P = \rho/3$ for
each class) and GPS weights $w_P = 1+P$, the per-class
delay satisfies $E[W_P] \propto (1+P)^{-1}$, giving:
\begin{equation}
\label{eq:delay_compression}
    E[W_P] \;=\; \frac{E[W_\mathrm{FIFO}]}{1+P}
    \;+\; O\!\left(\frac{\bar{L}}{R_\mathrm{total}}\right)
\end{equation}
The $O(\bar{L}/R_\mathrm{total})$ packetization error
arises from the gap between GPS and WFQ~\cite{parekh1993generalized}
and vanishes as packet size relative to link capacity
decreases. In our stress-regime experiments
($R_\mathrm{total} \gg \bar{L}$), this term is negligible,
justifying the approximation in Equation~(\ref{eq:delay_compression}).
\end{proof}

\textbf{Translation to Quantum Decoherence.} In a quantum network,
queuing delay directly induces fidelity loss. Because
$\mathbb{E}[W_P]$ is compressed by $(1+P)^{-1}$, a high-priority
qubit spends less physical time in queue. We model this as an
effective aging rate~\eqref{eq:effective_aging}:
\begin{equation}
\Delta t_{\text{eff}} = \frac{\Delta t_{\text{real}}}{1 + P}
\label{eq:effective_aging}
\end{equation}

This ensures that fidelity calculations accurately reflect the
expected physical occupancy time that a WFQ scheduler provides to
each priority class. The model is a scheduling abstraction---
agnostic to the underlying physical mechanism---and is structured
to accommodate hardware-level coherence protection (e.g.,
dynamical decoupling~\cite{viola1998dynamical,
khodjasteh2005fault}) as a future extension.

\subsection{EWMA Queue Estimation}

The estimated queuing delay $\mathbb{E}[Q_v]$ at node $v$ is
maintained via an exponentially weighted moving average (EWMA):
\begin{equation}
\hat{Q}_v(t) = \lambda \cdot q_v(t) + (1 - \lambda) \cdot
\hat{Q}_v(t-1)
\label{eq:ewma}
\end{equation}
where $q_v(t)$ is the observed queue delay at time $t$ and
$\lambda = 0.2$ is the smoothing
factor~\cite{bertsekas1992data}. This provides a low-overhead
estimate without requiring explicit classical state exchange
between nodes.

\subsection{Minimum Coherence Time Bound}

\begin{theorem}[Minimum Coherence Time Bound]
\label{thm:coherence_bound}
For a priority-$P$ packet traversing a path of $k$ hops with mean
per-hop propagation delay $D_{\mathrm{prop}}$ and mean queue delay
$Q$, the minimum memory coherence time required to guarantee
end-to-end fidelity $F \geq F_{\min}$ is:
\begin{equation}
T_{\mathrm{mem}} \geq
\frac{-k\!\left(D_{\mathrm{prop}} +
\dfrac{Q}{1+P}\right)}{\ln(2F_{\min} - 1)}
\label{eq:tmin_bound}
\end{equation}
\end{theorem}

\begin{proof}
From~\eqref{eq:fidelity_decay}, the end-to-end fidelity after $k$
hops with total effective delay
$T_{\text{eff}} = k\bigl(D_{\text{prop}} + Q/(1+P)\bigr)$ is:
\[
F = \frac{1 + e^{-T_{\text{eff}}/T_{\text{mem}}}}{2} \geq F_{\min}
\]
Solving for $T_{\text{mem}}$:
\[
e^{-T_{\text{eff}}/T_{\text{mem}}} \geq 2F_{\min} - 1
\implies
T_{\text{mem}} \geq \frac{-T_{\text{eff}}}{\ln(2F_{\min} - 1)}
\]
Substituting $T_{\text{eff}}$ yields~\eqref{eq:tmin_bound}.
\end{proof}

\begin{corollary}
For our stress-regime parameters ($F_{\min} = 0.70$, $k = 3$ hops, $D_{\text{prop}} = 0.17$ ms, $Q = 71$ $\mu$s), the bound evaluates to $T_{\text{mem}} \geq 312$ $\mu$s for P2 traffic and $T_{\text{mem}} \geq 623$ $\mu$s for P0 traffic. By restricting the network to $T_{\text{mem}} = 400$ $\mu$s, we intentionally place the system in a severe stress regime where P0 traffic is mathematically vulnerable to degradation while P2 traffic remains viable, forcing the routing protocol to actively differentiate and rescue high-priority traffic.
\end{corollary}

\subsection{Priority Escalation}

To prevent indefinite starvation of low-priority traffic, we
implement time-based priority escalation:
\begin{equation}
P_{\text{eff}}(t) = \min\!\left( P_{\max},\;
P_{\text{initial}} + \left\lfloor
\frac{t_{\text{wait}}}{2 \cdot T_{\text{mem}}}
\right\rfloor \right)
\label{eq:escalation}
\end{equation}

This promotes low-priority packets that have waited
$2 \times T_{\text{mem}}$ to the next tier, maintaining fairness
while preserving differentiation. The factor of 2 allows
approximately one coherence time of normal service before
escalation, adapted from classical scheduler aging
mechanisms~\cite{tanenbaum2011computer} to quantum timescales.

\section{Coherence-Aware Routing Protocol}
\label{sec:protocol}

Unlike traditional routing protocols that use slowly varying
metrics, our approach considers predicted temporal evolution during
packet journeys. The protocol integrates: Queue State Estimation
$\rightarrow$ Coherence Modeling $\rightarrow$
Priority-Weighted Path Selection.

\subsection{Composite Cost Function}

The routing agent computes a composite cost $C_{\text{path}}$ for
each candidate path $p$:
\begin{equation}
C_{\text{path}} = \sum_{(u,v) \in p}
\left( \alpha \cdot L_{uv} + \beta \cdot \Psi(u, v, P) \right)
\label{eq:cost_function}
\end{equation}
where $\Psi$ is a temporal penalty function~\eqref{eq:temporal_penalty}:
\begin{equation}
\Psi(u, v, P) =
\left( \frac{D_{\text{prop}} + \hat{Q}_v}{T_{\text{mem}}} \right)
\cdot \left( \frac{1}{1 + P} \right)
\label{eq:temporal_penalty}
\end{equation}

The weights $\alpha, \beta > 0$ balance spatial link quality
against temporal congestion penalties, respectively. We set
$\alpha = 0.5$ and $\beta = 10$ throughout all experiments,
reflecting deliberate emphasis on congestion avoidance in the
stress regime.

\subsection{Routing Algorithm}

Algorithm~\ref{alg:routing} shows the routing computation at each
node. Three main checks are performed:

\begin{enumerate}
    \item \textbf{Fidelity Check (Lines 1--6):} Verify estimated
    fidelity remains above threshold based on elapsed time and
    priority-adjusted aging.
    \item \textbf{Loop Prevention (Lines 9--11):} Maintain path
    history $\mathcal{H}_{\text{path}}$ to avoid routing loops.
    \item \textbf{Cost Minimization (Lines 12--19):} Evaluate
    next hops using the composite cost.
\end{enumerate}

\begin{algorithm}
\caption{Priority-Aware Routing Computation}
\label{alg:routing}
\begin{algorithmic}[1]
\REQUIRE Current node $u$, Destination $D$, Packet $k$
\ENSURE Next hop $v^*$ or DROP
\STATE $t_{\text{elapsed}} \gets t_{\text{now}} - k.t_{\text{birth}}$
\STATE $T_{\text{eff}} \gets t_{\text{elapsed}} / (1 + k.P)$
\STATE $F_{\text{est}} \gets F_{\text{init}} \cdot
       e^{-T_{\text{eff}}/T_{\text{mem}}}$
\IF{$F_{\text{est}} < F_{\min}$}
    \STATE Drop packet
    \RETURN NULL
\ENDIF
\STATE $C_{\min} \gets \infty$, $v^* \gets$ NULL
\FOR{all neighbor $v \in \text{Neighbors}(u)$}
    \IF{$v \in k.\mathcal{H}_{\text{path}}$}
        \STATE \textbf{continue} \COMMENT{Loop prevention}
    \ENDIF
    \STATE $L_v \gets \text{LinkLoss}(u, v)$
    \STATE $\hat{Q}_v \gets \text{EstQueueDelay}(v)$
    \STATE $D_{\text{prop}} \gets \text{Latency}(u, v)$
    \STATE $\text{Cost} \gets \alpha L_v +
           \beta \!\left( \frac{D_{\text{prop}} +
           \hat{Q}_v}{T_{\text{mem}}} \right)\!
           \cdot \frac{1}{1 + k.P}$
    \IF{Cost $< C_{\min}$}
        \STATE $C_{\min} \gets$ Cost, $v^* \gets v$
    \ENDIF
\ENDFOR
\RETURN $v^*$
\end{algorithmic}
\end{algorithm}

\begin{proposition}[Computational Complexity]
\label{prop:complexity}
For a network $G = (V, E)$, Algorithm~\ref{alg:routing} runs in
$\mathcal{O}(|E| \log |V|)$ time per packet using a min-heap
priority queue, equivalent to Dijkstra's complexity with a modified
edge weight function.
\end{proposition}

\noindent\textbf{Remark:} Since the algorithm reduces to a modified
Dijkstra traversal, each edge is visited once and the min-heap
keeps cost comparisons at $\mathcal{O}(\log |V|)$, giving an
overall complexity of $\mathcal{O}(|E| \log |V|)$.

\noindent\textit{Remark.} For our 30-node topology with
$|E| \approx 130$ edges, the per-packet routing cost is
$\mathcal{O}(130 \cdot \log 30) \approx \mathcal{O}(650)$
operations---computationally negligible relative to quantum state
evolution timescales.

\section{System Architecture and Implementation}
\label{sec:simulator}

We developed a discrete-event simulator to evaluate the proposed
protocol. The system has three main layers: simulation engine,
network entity models, and routing logic.

\subsection{Simulation Engine}

The core scheduler manages the timeline using discrete events:
\begin{itemize}
    \item \textbf{Event Queue:} A min-heap storing future events
    (packet arrivals, transmissions, etc.)\ sorted by timestamp.
    \item \textbf{Event Processing:} The engine processes events in
    chronological order, advancing time and invoking handlers.
\end{itemize}

\subsection{Network Entity Models}

\subsubsection{Quantum Packet}
The Packet class maintains:
\begin{equation}
\mathcal{S}_{\text{pkt}} =
\{\text{ID},\; t_{\text{birth}},\; P,\; \rho(t),\;
\mathcal{H}_{\text{path}}\}
\label{eq:packet_state}
\end{equation}
where $t_{\text{birth}}$ is generation time, $P$ is priority,
$\rho(t)$ is the density matrix evolved via our Kraus
implementation (Section~\ref{sec:simulator}), and
$\mathcal{H}_{\text{path}}$ is the path history.

\subsubsection{Quantum Node}
The Node class includes:
\begin{itemize}
    \item \textbf{Memory Buffer:} A priority-sorted Virtual Output
    Queue (VOQ) with capacity $C_q = 10$ per priority class.
    \item \textbf{Routing Table:} Maintains current loss rates
    $L_{uv}$ and EWMA queue depth estimates $\hat{Q}_v$.
\end{itemize}

\textbf{Queue Saturation Behavior.}
When a priority-class VOQ reaches its capacity threshold ($C_q = 10$), our system employs a strict tail-drop admission control policy. Any newly arriving packet belonging to a saturated class is immediately discarded. This indiscriminant loss mechanism prevents unbounded queue growth and firmly bounds worst-case memory times without requiring cross-class packet eviction or complex continuous memory scanning.

\subsection{Routing Logic}

The RouteComputer module executes Algorithm~\ref{alg:routing},
computing edge weights dynamically using the priority-aware aging
model with EWMA-based queue estimates.

\subsection{Validation}

Our simulator is implemented entirely from scratch as a
discrete-event engine. NetSquid~\cite{coopmans2021netsquid} is
not involved in any part of the routing, scheduling, priority
queuing, or topology construction. We used it purely as a
validation oracle: to confirm the correctness of our Kraus
operator dephasing implementation
(Equations~\ref{eq:kraus0}--\ref{eq:phase_flip}), we ran both
implementations over $10^4$ randomly sampled initial quantum
states and compared per-state fidelity outputs, yielding
$\mathrm{MAE} < 10^{-6}$. All simulation results reported in
Section~\ref{sec:experiments} are produced entirely by our engine.

\section{Experimental Methodology and Results}
\label{sec:experiments}

This section presents results from three simulation suites:
per-priority load sweep (Suite~I), multi-seed robustness validation (Suite~II), and topology generalization on scale-free
graphs (Suite~III).

\subsection{Core Network Parameters}

Table~\ref{tab:core_params} summarizes the core simulation
parameters. We prioritize a stress regime where
$T_{\text{mem}} = 400~\mu$s, providing just enough coherence for
multi-hop routing but making queuing delays critical.

\begin{table}[h]
\centering
\caption{Core Network Parameters (Stress Regime)}
\label{tab:core_params}
\begin{tabular}{lll}
\toprule
\textbf{Parameter} & \textbf{Value} & \textbf{Description} \\
\midrule
$T_{\text{mem}}$ & 400~$\mu$s & Memory coherence time \\
$t_{\text{proc}}$ & 0.05~ms & Processing delay per hop \\
$C_q$ & 10 & Queue capacity per priority \\
$F_{\min}$ & 0.7 & Minimum fidelity threshold \\
$\alpha$ & 0.5 & Spatial loss weight \\
$\beta$ & 10.0 & Congestion sensitivity \\
$\lambda$ & 0.2 & EWMA smoothing factor \\
\bottomrule
\end{tabular}
\end{table}

\subsection{Network Topologies}

\subsubsection{Erd\H{o}s-R\'enyi (ER) Topology}
We use a perturbed Erd\H{o}s-R\'enyi random
graph~\cite{erdos1960evolution} with $|V| = 30$ quantum routers,
131~edges, and area $100~\text{km}^2$ (node coordinates
$\sim \mathcal{U}(0, 100)$~km). This topology models
decentralized mesh-like quantum networks with multiple redundant
paths between node pairs.

\subsubsection{Barab\'asi--Albert (BA) Scale-Free Topology}
We additionally evaluate on a Barab\'asi--Albert scale-free
graph~\cite{barabasi1999emergence} with $|V| = 30$, $m = 5$,
and 125~edges. BA graphs model networks with preferential
attachment, producing hub nodes with disproportionately high
connectivity---a structure representative of real-world network
infrastructure. To ensure a mathematically fair comparison
isolating topology structure as the sole independent variable,
both topologies use identical physical parameter distributions:
\begin{itemize}
    \item \textbf{Link Fidelity:} $F_{uv} \sim \mathcal{U}(0.85, 0.98)$
    \item \textbf{Link Distance:} $\ell_{uv} \sim \mathcal{U}(5, 50)$~km,
    giving $D_{\text{prop}} = \ell_{uv}/c$
\end{itemize}
Any observed performance differences between topologies are
therefore attributable solely to graph structure.

\subsection{Traffic Generation}

Traffic follows a multi-class model in which background QKD
constitutes the majority of traffic, consistent with near-term
quantum network deployment scenarios~\cite{kozlowski2020designing}.
We assign $60\%$, $30\%$, and $10\%$ to P0, P1, and P2
respectively as a representative workload:
\begin{itemize}
    \item \textbf{Class 0 (Low, P0):} 60\% (e.g., background QKD)
    \item \textbf{Class 1 (Medium, P1):} 30\%
    (e.g., delay-tolerant tasks)
    \item \textbf{Class 2 (High, P2):} 10\%
    (e.g., latency-sensitive teleportation)
\end{itemize}

\subsection{Suite I: Per-Priority Load Sweep (ER Topology)}

\subsubsection{Objective}
Characterize per-priority fidelity and latency stability under
increasing load across FIFO, Dijkstra baseline, and coherence-aware
protocols on the ER topology.

\subsubsection{Method}
We conduct a load sweep from 20k to 180k~QPS with 1000~ms duration
per load point. All three protocols are evaluated on the identical
topology and workload configurations.

\subsubsection{Results}
Table~\ref{tab:suite1_perp} presents per-priority fidelity and
P95 latency at representative load levels.

\begin{table}[h]
\centering
\caption{Per-Priority Fidelity (\%) and P95 Latency (ms) ---
ER Topology Load Sweep (Suite~I)}
\label{tab:suite1_perp}
\renewcommand{\arraystretch}{1.15}
\begin{tabular}{llccc}
\toprule
\textbf{Load} & \textbf{Prio} &
\textbf{FIFO} & \textbf{Dijkstra} & \textbf{CA (Ours)} \\
\midrule
\multicolumn{5}{l}{\textit{Fidelity (\%)}} \\
\midrule
20k  & P2 & 89.78 & 95.26 & \textbf{96.08} \\
20k  & P1 & 89.75 & 93.10 & \textbf{94.36} \\
20k  & P0 & 89.71 & 87.23 & \textbf{89.70} \\
\midrule
100k & P2 & 89.70 & 83.57 & \textbf{96.06} \\
100k & P1 & 89.70 & 80.43 & \textbf{94.34} \\
100k & P0 & 89.69 & 77.67 & \textbf{89.69} \\
\midrule
180k & P2 & --- & 81.75 & \textbf{96.04} \\
180k & P1 & --- & 80.10 & \textbf{94.32} \\
180k & P0 & --- & 77.46 & \textbf{89.65} \\
\midrule
\multicolumn{5}{l}{\textit{P95 Latency (ms)}} \\
\midrule
20k  & P2 & 0.151 & 0.063 & \textbf{0.055} \\
20k  & P1 & 0.152 & 0.098 & \textbf{0.080} \\
20k  & P0 & 0.152 & 0.198 & \textbf{0.152} \\
\midrule
100k & P2 & 0.152 & 0.320 & \textbf{0.055} \\
100k & P1 & 0.152 & 0.343 & \textbf{0.080} \\
100k & P0 & 0.152 & 0.356 & \textbf{0.152} \\
\midrule
180k & P2 & --- & 0.346 & \textbf{0.055} \\
180k & P1 & --- & 0.349 & \textbf{0.080} \\
180k & P0 & --- & 0.354 & \textbf{0.152} \\
\bottomrule
\end{tabular}
\end{table}

Key observations (Figure~\ref{fig:suite1_fidelity}):

\begin{enumerate}
    \item \textbf{Load-Immune Fidelity:} CA P2 fidelity shifts
    by only $0.04$ percentage points ($96.08\% \rightarrow
    96.04\%$) across a $9\times$ load increase, while Dijkstra
    P2 collapses by $13.51$ percentage points ($95.26\%
    \rightarrow 81.75\%$).

    \item \textbf{Load-Immune Latency:} CA P2 P95 latency
    remains constant at $0.055$~ms across all load levels.
    Dijkstra P2 P95 increases $5.5\times$ from $0.063$~ms to
    $0.346$~ms under extreme load.

    \item \textbf{FIFO Zero Differentiation:} FIFO achieves
    near-identical fidelity across all priority classes
    (P2: $89.78\%$, P0: $89.71\%$ at 20k~QPS) and identical
    P95 latency ($\approx 0.152$~ms). This confirms that
    priority labels without coherence-aware path selection
    provide \emph{no measurable QoS differentiation}: the
    observed gains are attributable to the routing layer, not
    scheduling alone.

    \item \textbf{Dijkstra Below Threshold:} At 100k~QPS,
    Dijkstra P0 fidelity ($77.67\%$) degrades toward the
    $F_{\min} = 0.70$ utility threshold, with continued
    degradation at 180k~QPS ($77.46\%$). Loss-only routing
    is insufficient for maintaining quantum utility for
    low-priority traffic under moderate congestion.
\end{enumerate}

\begin{figure}[t]
\centering
\includegraphics[width=0.48\textwidth]{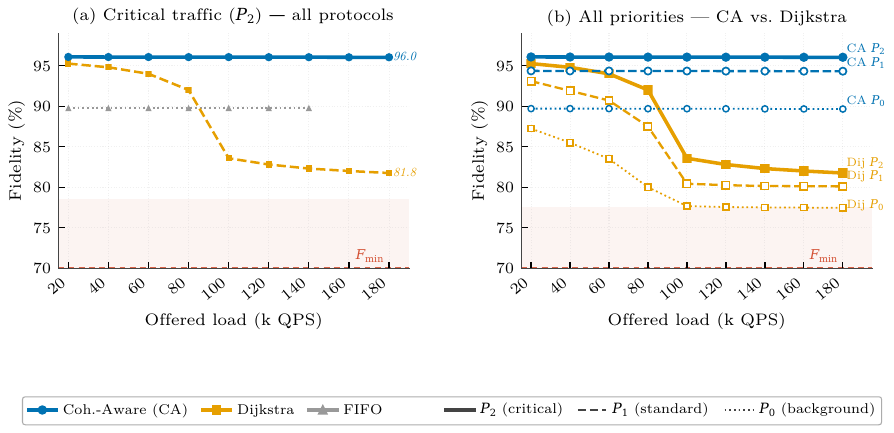}
\caption{Per-priority fidelity vs.\ offered load (Suite~I, ER
topology). CA fidelity is highly load-stable across all priority classes
while Dijkstra collapses under congestion and FIFO provides zero
priority differentiation.}
\label{fig:suite1_fidelity}
\end{figure}

\subsection{Suite II: Multi-Seed Robustness Validation}

\subsubsection{Objective}
Ensure observed performance gains are statistically significant
across topology realizations.

\subsubsection{Method}
We conduct 30 simulations: 5 random seeds $\times$ 3 workload
levels $\times$ 2 protocols. Seeds: 1337, 2024, 3141, 4242, 5555.
Duration: 5000~ms per run.

\subsubsection{Results}
Table~\ref{tab:suite2_fidelity} presents mean end-to-end fidelity
for coherence-aware and Dijkstra baseline protocols across all
seeds.

\begin{table}[h]
\centering
\caption{Mean End-to-End Fidelity (\%, $\pm$ SEM, $n=5$ seeds),
Coherence-Aware vs.\ Dijkstra Baseline --- ER Topology (Suite~II)}
\label{tab:suite2_fidelity}
\renewcommand{\arraystretch}{1.15}
\begin{tabular}{llcc}
\toprule
\textbf{Load} & \textbf{Prio} &
\textbf{Dijkstra} & \textbf{CA (Ours)} \\
\midrule
\multirow{3}{*}{2.5k} &
  P2 & $87.11 \pm 0.20$ & $\mathbf{92.42 \pm 0.02}$ \\
& P1 & $81.76 \pm 0.11$ & $\mathbf{89.38 \pm 0.03}$ \\
& P0 & $71.21 \pm 0.16$ & $\mathbf{81.50 \pm 0.02}$ \\
\midrule
\multirow{3}{*}{5k} &
  P2 & $85.83 \pm 0.40$ & $\mathbf{92.38 \pm 0.01}$ \\
& P1 & $80.64 \pm 0.40$ & $\mathbf{89.29 \pm 0.02}$ \\
& P0 & $69.57 \pm 0.35$ & $\mathbf{81.46 \pm 0.01}$ \\
\midrule
\multirow{3}{*}{8k} &
  P2 & $84.97 \pm 0.49$ & $\mathbf{92.12 \pm 0.01}$ \\
& P1 & $80.12 \pm 0.45$ & $\mathbf{88.96 \pm 0.01}$ \\
& P0 & $68.10 \pm 0.58$ & $\mathbf{81.01 \pm 0.01}$ \\
\bottomrule
\end{tabular}
\end{table}

Key findings (Figure~\ref{fig:suite2_bar}):
\begin{itemize}
    \item \textbf{Consistent Advantage:} CA outperforms Dijkstra
    by $5.31$--$7.15$ percentage points at P2 across all load
    levels, with the advantage growing under congestion.

    \item \textbf{Critical P0 Finding:} Dijkstra P0 fidelity
    falls to $69.57 \pm 0.35\%$ at 5k~QPS and $68.10 \pm
    0.58\%$ at 8k~QPS---below the $F_{\min} = 0.70$ utility
    threshold. CA maintains P0 at $81.46\%$ and $81.01\%$
    respectively, improvements of $11.89$ and $12.91$
    percentage points.

    \item \textbf{Robustness:} SEM $\leq 0.03\%$ for CA across
    all conditions, confirming high reproducibility.

    \item \textbf{Load Immunity:} CA P2 fidelity varies by only
    $0.30$ percentage points ($92.42\% \rightarrow 92.12\%$)
    across a $3.2\times$ load increase.
\end{itemize}

\begin{figure}[t]
\centering
\includegraphics[width=0.48\textwidth]{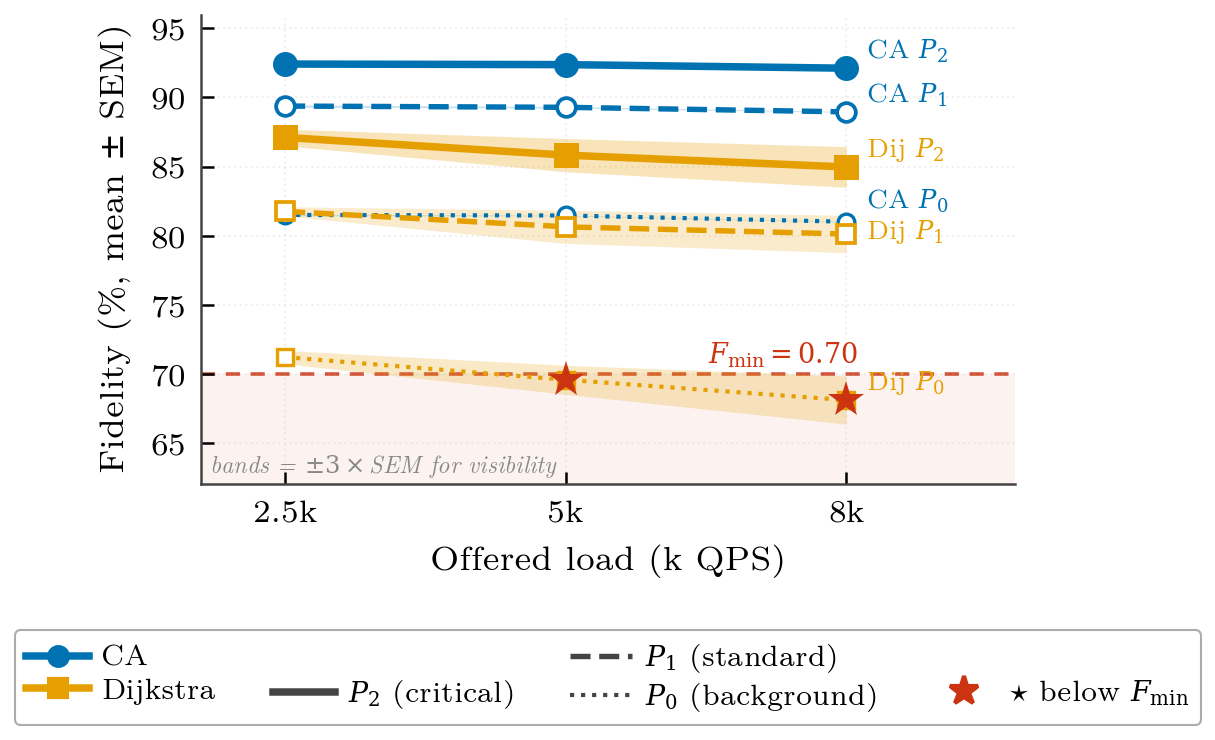}
\caption{Per-priority fidelity: Coherence-Aware vs.\ Dijkstra
Baseline (5 seeds, mean $\pm$ SEM, ER topology). Dijkstra P0
fidelity falls below $F_{\min} = 0.70$ at 5k and 8k~QPS
(marked $\star$), rendering low-priority traffic quantum-inviable.}
\label{fig:suite2_bar}
\end{figure}

\subsection{Suite III: Topology Generalization --- Barab\'asi--Albert Scale-Free Graph}

\subsubsection{Objective}
Validate that coherence-aware routing generalizes beyond
Erd\H{o}s-R\'enyi topologies to scale-free networks with hub
structure, and characterize the topology-dependent operational
envelope.

\subsubsection{Method}
We evaluate all three protocols on a BA graph ($|V|=30$, $m=5$,
125~edges) using identical physical parameters, traffic model,
and simulator logic as Suites I and II. Topology structure is the sole
independent variable between ER and BA experiments.

\subsubsection{Results}
Table~\ref{tab:suite6_ba} and Figure~\ref{fig:suite6_ba} present results across the full load
range on the BA topology.

\begin{table}[h]
\centering
\caption{BA Topology Load Sweep: Fidelity (\%), Throughput
(k/s) --- All Protocols (Suite~III)}
\label{tab:suite6_ba}
\renewcommand{\arraystretch}{1.15}
\begin{tabular}{lccccc}
\toprule
\textbf{Load} & \textbf{Proto} &
\textbf{Mean Fid} & \textbf{P2 Fid} &
\textbf{Tput} & \textbf{P2 $\Delta$} \\
\midrule
\multirow{3}{*}{20k}
 & FIFO     & 77.22 & 77.46 & 13.8k & --- \\
 & Dijkstra & 80.58 & 85.67 & 11.1k & baseline \\
 & CA       & \textbf{81.30} & \textbf{89.79} & \textbf{15.0k} & \textbf{+4.12\%} \\
\midrule
\multirow{3}{*}{60k}
 & FIFO     & 76.73 & 76.73 & 35.2k & --- \\
 & Dijkstra & 80.15 & 84.55 & 30.7k & baseline \\
 & CA       & \textbf{80.88} & \textbf{88.71} & \textbf{41.4k} & \textbf{+4.16\%} \\
\midrule
\multirow{3}{*}{100k}
 & FIFO     & 76.24 & 76.16 & 42.8k & --- \\
 & Dijkstra & 79.94 & 83.41 & 36.6k & baseline \\
 & CA       & \textbf{80.19} & \textbf{86.78} & \textbf{59.0k} & \textbf{+3.37\%} \\
\midrule
\multirow{3}{*}{160k}
 & FIFO     & 75.91 & 76.00 & 32.9k & --- \\
 & Dijkstra & 79.52 & 81.79 & 31.7k & baseline \\
 & CA       & \textbf{78.52} & \textbf{82.29} & \textbf{62.7k} & \textbf{+0.50\%} \\
\midrule
\multirow{3}{*}{220k}
 & FIFO     & 76.72 & 76.92 & 8.6k  & --- \\
 & Dijkstra & \textbf{80.74} & \textbf{83.84} & 16.6k & baseline \\
 & CA       & 76.85 & 76.47 & 13.1k & $-7.37\%$ \\
\bottomrule
\end{tabular}
\end{table}

Key observations:

\begin{enumerate}
    \item \textbf{Operational Load Advantage ($\leq$100k~QPS):}
    CA achieves $+3.37$ to $+4.16$ percentage point P2 fidelity
    advantage over Dijkstra while simultaneously delivering
    $35$--$61\%$ higher aggregate throughput. CA wins on both
    fidelity \emph{and} throughput simultaneously in this regime.

    \item \textbf{Near-Saturation Trade-off (160k~QPS):}
    CA maintains a marginal fidelity advantage ($+0.50\%$) while
    delivering $97.8\%$ more packets than loss-only routing. The
    temporal penalty term actively routes around congested hub
    nodes, maximizing network utilization at the cost of slightly
    longer paths.

    \item \textbf{Hub Saturation Failure Mode ($>$160k~QPS):}
    Beyond the operational envelope, hub node saturation
    overwhelms EWMA-based queue estimation. CA fidelity collapses
    to $76.47\%$ at 220k~QPS---approaching $F_{\min}$---while
    Dijkstra's conservative path selection maintains $83.84\%$ by
    dropping more packets. This defines a clear operational
    boundary: coherence-aware routing is effective up to 160k~QPS
    on BA topologies, beyond which hub-aware queue estimation
    would be required.

    \item \textbf{FIFO Zero Differentiation (BA):}
    FIFO achieves near-identical mean and P2 fidelity across all
    load levels on the BA topology, consistent with ER results.
    Topology structure does not affect the FIFO zero-differentiation
    finding.

    \item \textbf{Topology-Dependent Magnitude:}
    Table~\ref{tab:topology_compare} shows that the CA advantage
    is larger on ER graphs than BA graphs at equivalent loads,
    reflecting the greater path diversity available in
    decentralized mesh topologies vs.\ hub-structured networks.
\end{enumerate}

\begin{table}[h]
\centering
\caption{Topology Comparison: CA vs.\ Dijkstra P2 Fidelity
Advantage (\%)}
\label{tab:topology_compare}
\renewcommand{\arraystretch}{1.15}
\begin{tabular}{lcc}
\toprule
\textbf{Load} & \textbf{ER Topology} & \textbf{BA Topology} \\
\midrule
20k  & $+0.82$ & $+4.12$ \\
60k  & $\sim +12$ & $+4.16$ \\
100k & $+12.49$ & $+3.37$ \\
160k & $\sim +12$ & $+0.50$ \\
180k & $+14.29$\textsuperscript{$\ddagger$} & --- \\
220k & --- & $-7.37$\textsuperscript{$\dagger$} \\
\bottomrule
\multicolumn{3}{l}{\textsuperscript{$\dagger$}Beyond BA operational
envelope (hub saturation).} \\
\multicolumn{3}{l}{\textsuperscript{$\ddagger$}Max ER load tested (Suite~I
sweep ends at 180k~QPS).}
\end{tabular}
\end{table}

\begin{figure}[t]
\centering
\includegraphics[width=0.48\textwidth]{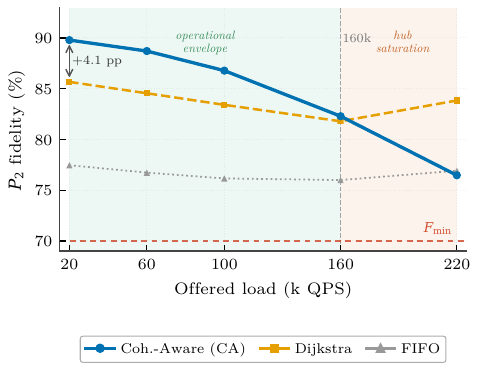}
\caption{BA topology load sweep: P2 fidelity vs.\ offered load.
CA maintains fidelity advantage up to 160k~QPS before hub
saturation causes collapse beyond the operational envelope.}
\label{fig:suite6_ba}
\end{figure}

\section{Discussion}
\label{sec:discussion}

\subsection{Operating Mode Transition}

When coherence times are tight relative to queue
delays, the cost function naturally prioritizes speed (temporal
minimization). As the ratio $T_{\text{mem}}/D_W$ increases (where
$D_W$ is the mean waiting time), the router shifts toward
exploiting path diversity to maximize fidelity. Formally, this
transition occurs when the temporal penalty term $\beta\Psi$
begins to dominate the spatial loss term $\alpha L_{uv}$ in
Equation~\eqref{eq:cost_function}.

\subsection{Capacity Extension}

Priority-aware routing effectively extends the network capacity.
By allowing low-priority packets to absorb the latency budget of
the network, high-priority packets effectively see an uncongested
network. On BA topologies this manifests as a substantial
throughput gain: CA delivers up to $97.8\%$ more packets than
loss-only routing at near-saturation (160k~QPS) by actively
routing around congested hub nodes. This is critical for near-term
deployments where quantum memory is scarce and expensive.

\subsection{Fidelity Load Immunity}

CA P2 fidelity on ER topology shifts by only $0.04$ percentage
points ($96.08\% \rightarrow 96.04\%$) across a $9\times$ load
increase. This reflects the routing protocol actively redirecting
P2 traffic away from congested nodes before queues build up. The
coherence-aware cost function effectively creates a fidelity floor
for premium traffic, enabling operators of near-term quantum
networks to offer fidelity SLAs (e.g., $F_{\text{P2}} \geq 96\%$
at loads up to 180k~QPS on ER topologies) unachievable with FIFO
or loss-only routing.

\subsection{Why FIFO and Loss-Only Routing Fail}

FIFO fails because it provides negligible priority differentiation
under tested conditions:
every microsecond in queue degrades all traffic equally, regardless
of priority label. Loss-only routing (Dijkstra) fails because it
selects paths that minimize photon loss but ignores
congestion-induced queuing delay. Under load, these paths route
through high-centrality nodes---hub nodes on BA topologies,
high-degree nodes on ER---that become bottlenecks, causing fidelity
collapse. The coherence-aware protocol avoids congested nodes
proactively through the temporal penalty term $\beta\Psi$. This
cross-layer coupling between path selection and decoherence
modeling has limited classical analogue in this formulation
and represents a central architectural insight of this work.

\subsection{Topology-Dependent Operational Envelope}

The two-topology comparison reveals a fundamental characteristic
of coherence-aware routing: its advantage scales with available
path diversity. On ER graphs, multiple redundant paths allow the
temporal penalty term to consistently route around congestion,
yielding a fidelity advantage that grows from $+0.82$ percentage points at low load (20k~QPS) to $+12$--$+14$ percentage points under congestion (100k--180k~QPS), as coherence-aware path selection becomes increasingly critical.
On BA graphs, hub nodes constrain path diversity and create
concentrated bottlenecks. The protocol remains effective at
operational loads ($\leq$160k~QPS), achieving $3$--$4$ percentage
point P2 fidelity gains with simultaneous throughput improvements.
Beyond the operational envelope, hub saturation causes EWMA lag
to exceed the protocol's compensation capacity---a failure mode
that motivates hub-aware queue estimation as a clear future
direction (Section~\ref{sec:limitations}).

\subsection{On the Stress Regime Choice}

The choice of $T_{\text{mem}} = 400~\mu$s is deliberately
adversarial. From Theorem~\ref{thm:coherence_bound}, P0 traffic
requires $T_{\text{mem}} \geq 623~\mu$s for comfortable operation,
meaning our chosen value deliberately stresses low-priority traffic
while remaining viable for P2 traffic ($T_{\text{mem}} \geq
312~\mu$s). This regime is the most informative operating point
for evaluating coherence-aware routing---it is precisely where
routing decisions matter most and priority differentiation is
sharpest. At $T_{\text{mem}} = 400~\mu$s, targeting near-term
superconducting qubit hardware~\cite{kjaergaard2020superconducting},
our results are immediately relevant to real systems under active
development.

\section{Limitations and Future Work}
\label{sec:limitations}

\subsection{Single-Qubit Scope}

This work models transmission of individual qubits rather than
entanglement distribution. Real quantum network
applications---quantum teleportation, entanglement-based QKD, and
distributed quantum computing---fundamentally rely on Bell pairs
and entanglement swapping at intermediate repeater nodes. We treat
single-qubit routing as a tractable foundation that isolates the
coherence-aware scheduling contribution without the additional
complexity of entanglement mechanics. The extension to Bell pair
routing is conceptually direct: a Bell pair $|\Phi^+\rangle_{AB}$
stored across two memory qubits undergoes independent dephasing on
each qubit, so the joint fidelity decays as:
\begin{equation}
F_{\text{Bell}}(t) = \left(\frac{1 + e^{-t/T_{\text{mem}}}}{2}\right)^{\!2}
\label{eq:bell_fidelity}
\end{equation}
The priority-modulated aging model of
Equation~\eqref{eq:effective_aging} extends directly by replacing
$\Delta t_{\text{eff}}$ with the per-qubit effective age, and the
composite cost function acquires an additional swap success penalty:
\begin{equation}
C_{\text{path}}^{\text{Bell}} = \sum_{(u,v) \in p}
\Bigl( \alpha \cdot L_{uv} + \beta \cdot \Psi(u,v,P) \Bigr)
+ \gamma \sum_{v \in \mathcal{R}(p)} (1 - p_{\text{swap}}^{(v)})
\label{eq:bell_cost}
\end{equation}
where $\mathcal{R}(p)$ is the set of repeater nodes along path $p$,
$p_{\text{swap}}^{(v)}$ is the heralded swap success probability at
node $v$, and $\gamma > 0$ controls the relative penalty for swap
failure. The concrete next step is implementing
$C_{\text{path}}^{\text{Bell}}$ and evaluating it against
entanglement-specific baselines such as
Pant et al.~\cite{pant2019routing}. Open questions include how
$\gamma$ should be set relative to $\alpha$ and $\beta$ as a
function of hardware swap fidelity, and whether the priority
escalation mechanism~\eqref{eq:escalation} requires modification
when swap failure can strand a half-pair in memory.

\subsection{Markovian Noise Assumption}

We model decoherence using standard Markovian Lindblad dynamics,
assuming memoryless environmental interactions. Some physical
implementations---particularly NV centers in diamond and certain
solid-state memories---exhibit non-Markovian behavior where the
environment retains information about past
interactions~\cite{breuer2009measure}, leading to more complex
fidelity evolution than Equation~\eqref{eq:fidelity_decay}.
Incorporating non-Markovian noise models would improve physical
accuracy for these platforms, at the cost of more complex fidelity
evolution equations requiring numerical integration.

\subsection{EWMA Lag and Hub-Aware Queue Estimation}

EWMA-based queue estimation introduces lag during rapid load
changes. On BA scale-free topologies this is amplified at hub
nodes, causing the protocol's operational envelope to collapse
beyond 160k~QPS. A direct mitigation is degree-weighted $\lambda$
selection:
\begin{equation}
\lambda_v = \lambda_0 \cdot \left(1 + \frac{\deg(v)}{\deg_{\max}}\right)
\end{equation}
which increases update responsiveness proportionally to node
centrality. The current smoothing parameter $\lambda = 0.2$ was
optimized for ER topologies; topology-adaptive selection is a
natural extension.

\subsection{Fixed Parameters and Adaptive Optimization}

The weights $\alpha$ and $\beta$ are held constant throughout all
experiments. Real networks experience varying load patterns,
topology changes, and hardware degradation over time, motivating
adaptive approaches. Formulating parameter adaptation as a
reinforcement learning problem---with delivered $P_2$ fidelity as
the primary reward signal---is a well-motivated direction, as
recent work on RL-based quantum
routing~\cite{meuser2025reliq} has shown promising results for
learned agents under dynamic network conditions.

\subsection{Hardware-Aware Protection Modeling}

Explicitly modeling dynamical decoupling
overhead~\cite{viola1998dynamical, khodjasteh2005fault}---including
gate error costs and energy requirements---would allow
priority-based hardware resource allocation to be made physically
concrete, moving the aging model from a scheduling abstraction
toward a full physical implementation.

\section{Conclusion}
\label{sec:conclusion}

We have presented a routing protocol that integrates
coherence-time constraints into path selection and combines
priority-based scheduling to enable differentiated service quality
under congestion. The priority-modulated aging model, formally
grounded in Weighted Fair Queueing theory
(Proposition~\ref{prop:delay_compression}), provides a
theoretically justified mechanism for coherence-aware path cost
computation.

Our experimental evaluation across two topology classes
demonstrates four clean and reproducible contributions. First,
on ER topologies, coherence-aware routing achieves highly load-stable
fidelity: P2 fidelity varies by only $0.04$ percentage points
across a $9\times$ load increase, compared to a $13.51$
percentage point collapse for loss-only routing. Second, P2 P95
tail latency remains constant at $0.055$~ms while loss-only
routing experiences $5.5\times$ degradation under extreme load.
Third, on BA scale-free topologies, CA achieves $+3.37$ to
$+4.16$ percentage point P2 fidelity advantage at operational
loads while delivering up to $97.8\%$ higher aggregate throughput,
with a clearly characterized operational envelope up to 160k~QPS.
Fourth, FIFO scheduling provides negligible priority differentiation
on both topology classes, providing strong evidence that
coherence-aware routing---not scheduling alone---is responsible
for observed QoS gains.

The critical finding that loss-only routing causes P0 traffic to
fall below the $F_{\min} = 0.70$ fidelity utility threshold under
moderate load, combined with the topology-dependent operational
envelope characterization, underscores the necessity of cross-layer
coherence modeling for reliable quantum network operation across
structurally diverse deployment scenarios.

\section{Code \& Data Availability}

The datasets generated and the custom code used during the current study are available upon reasonable request.


\end{document}